\documentclass[11pt]{article}

\usepackage{wrapfig}
\usepackage{amsmath}
\usepackage{amssymb,amsfonts,textcomp}
\usepackage{latexsym}
\usepackage[noend,ruled]{algorithm2e} 

\usepackage{graphicx}
\usepackage{latexsym}
\usepackage{multirow}
\usepackage{rotating}

\newcommand{\reals}{{{\mathbb{R}}}}

\newcommand{\OPT}{{{{\mathrm{OPT}}}}}
\newcommand{\opt}{{{{\mathrm{opt}}}}}

\newcommand{\probability}{{{\mathrm{P}}}}

\newcommand{\Ev}{{\mathit{Ev}}}
\newcommand{\partialsol}{{\mathit{partial}}}
\newcommand{\best}{{\mathit{best}}}
\newcommand{\sol}{{\mathit{sol}}}
\newcommand{\hit}{{\mathit{hit}}}
\newcommand{\poly}{{\mathrm{poly}}}

\newcommand{\argmax}{{{{\mathrm{argmax}}}}}

\newcommand{\quality}{{{\mathit{quality}}}}
\newcommand{\Cov}{{{\mathit{Cov}}}}
\newcommand{\uCov}{{{\mathit{uCov}}}}

\title{Approximating the MaxCover Problem with Bounded Frequencies in
  FPT Time}

\author{Piotr Skowron\\ 
        University of Warsaw\\
        Warsaw, Poland\\
        \and
        Piotr Faliszewski\\
        AGH University\\
        Krakow, Poland
}

\newtheorem{theorem}{Theorem}

\newtheorem{corollary}[theorem]{Corollary}
\newtheorem{definition}{Definition}
\newtheorem{proposition}[theorem]{Proposition}

\newenvironment{proof}{\paragraph{Proof}}{\hfill$\Box$\medskip}

\newcommand{\np}{{\mathrm{NP}}}
\newcommand{\fpt}{{\mathrm{FPT}}}

\newcommand{\wone}{{\mathrm{W[1]}}}
\newcommand{\wtwo}{{\mathrm{W[2]}}}
\newcommand{\wpclass}{{\mathrm{W[P]}}}
\newcommand{\p}{{\mathrm{P}}}

\newcommand{\calI}{{{\mathcal{I}}}}
\newcommand{\calC}{{{\mathcal{C}}}}
\newcommand{\calA}{{{\mathcal{A}}}}

\newcommand{\calP}{{{\mathcal{P}}}}

\newcommand{\calS}{{{\mathcal{S}}}}

\allowdisplaybreaks

\begin{document}

\maketitle

\begin{abstract}
  We study approximation algorithms for several variants of the
  MaxCover problem, with the focus on algorithms that run in FPT time.
  In the MaxCover problem we are given a set $N$ of elements, a family
  $\calS$ of subsets of $N$, and an integer $K$. The goal is to find
  up to $K$ sets from $\calS$ that jointly cover (i.e., include) as
  many elements as possible. This problem is well-known to be
  $\np$-hard and, under standard complexity-theoretic assumptions, the
  best possible polynomial-time approximation algorithm has
  approximation ratio $(1 - \frac{1}{e})$.  We first consider a variant
  of MaxCover with bounded element frequencies, i.e., a variant where
  there is a constant $p$ such that each element belongs to at most
  $p$ sets in $\calS$.  For this case we show that there is an FPT
  approximation scheme (i.e., for each $\beta$ there is a $\beta$-approximation
  algorithm running in FPT time) for the problem of maximizing the number of
  covered elements, and a randomized FPT approximation scheme for the
  problem of minimizing the number of elements left uncovered (we take
  $K$ to be the parameter). Then, for the case where there is a
  constant $p$ such that each element belongs to at least $p$ sets
  from $\calS$, we show that the standard greedy approximation
  algorithm achieves approximation ratio exactly
  ${1-e^{-\max(pK/\|\calS\|, 1)}}$. We conclude by considering an unrestricted
  variant of MaxCover, and show approximation algorithms that run in
  exponential time and combine an exact algorithm with a greedy
  approximation. Some of our results improve currently known results for
  MaxVertexCover.
\end{abstract}

\section{Introduction}

We study approximation algorithms for, and parametrized complexity of,
the MaxCover problem with bounded frequency of the elements. In the
MaxCover problem we are given a set $N$ of $n$ elements, a family
$\calS = \{S_1, \ldots, S_m\}$ of $m$ subsets of $N$, and an integer
$K$. The goal is to find a size-at-most-$K$ subcollection of $\calS$
that covers as many elements from $N$ as possible. In the variant with
bounded frequencies of elements we further assume that there is some
constant $p$ such that each element appears in at most $p$ sets. A
particularly well-known special case of MaxCover with frequencies
upper-bounded by $2$ is the MaxVertexCover problem: We are given a
graph $G = (V,E)$ and the goal is to find $K$ vertices that, jointly,
are incident to as many edges as possible (i.e., the edges are the
elements to be covered and the vertices are the sets; clearly, each
edge ``belongs to'' exactly two vertices). Nonetheless, even for the
frequency upper bound $2$, MaxCover is considerably more general than
MaxVertexCover (e.g., the former allows two sets to have more than one
element in common, which is impossible in the latter\footnote{This
  difference may not sound particularly significant, but due to it
  some algorithms for MaxVertexCover (e.g., an FPT approximation
  scheme of Marx~\cite{Marx06parameterizedcomplexity}) do not
  generalize easily to the MaxCover setting.}).
%
%
%
In addition to MaxCover with upper-bounded frequencies, we also study
a variant of the problem with lower-bounded frequencies, and the
general variant, without any restrictions on element frequencies.

Our paper differs from the typical approach to the design of
approximation algorithms in that we do not focus on polynomial-time
algorithms, but also consider exponential-time ones.  For example, we
are interested in FPT approximation schemes, that is, in approximation
algorithms that for each desired approximation ratio $\beta$ output a
$\beta$-approximate solution in exponential time, but where the
exponential growth is only with respect to the number $K$ of sets that
we allow in the solution (and where $\beta$ is considered to be a
constant when computing the running time). In that respect, our work
is very close in spirit to the recent study of Croce and
Paschos~\cite{cro-pas:j:cover}, who---among other results---give
moderately exponential time (but not FPT-time) approximation schemes
for the MaxVertexCover problem. (However, there is also an FPT-time
approximation scheme for MaxVertexCover due to
Marx~\cite{Marx06parameterizedcomplexity}.)  Such exponential-time
approximation algorithms are desirable because they can achieve much
better approximation ratios than the polynomial-time ones, while still
being significantly faster than the currently-known exact
algorithms. We give more detailed review of related work in
Section~\ref{sec:related} and below we briefly describe our findings
and the motivation behind our research.

We obtain the following results (unless we mention otherwise, we
always consider our problems to be parametrized by $K$, the number of
the sets allowed in the solution).  First, building on the approach of
Guo et al.~\cite{guo-nie-wer:j:vertex-cover-variants}, in
Section~\ref{sec:worst-case} we show that the MaxCover problem with
bounded frequencies is $\wone$-complete.
On the other hand, without the frequency upper-bound assumption,
MaxCover is $\wtwo$-hard and we show that it belongs to $\wpclass$. We
also consider several other parameters and, in particular, we show
that MaxCover is $\wtwo$-complete for the parameter that combines the
number of sets we can use in the solution and the number of elements
that we are allowed to leave uncovered.  The core of the paper is,
however, in Section~\ref{sec:approximation}.  There, we show that for
each $\beta$, $0 < \beta < 1$, there is an FPT $\beta$-approximation
algorithm for the MaxCover problem with bounded frequencies.  On the
other hand, for the case where each element appears in \emph{at least}
$p$ out of $m$ sets, we show that the standard MaxCover greedy
approximation algorithm (i.e., one that picks one-by-one those sets
that include most not-yet-covered elements) achieves approximation
ratio $1-e^{-\frac{pK}{m}}$ (for the general case, this algorithm's
approximation ratio is $1-\frac{1}{e}$). Finally, we consider a
variant of the MaxCover problem where instead of maximizing the number
of covered elements, we minimize the number of those that remain
uncovered. We refer to this problem as the MinNonCovered problem.
Under the assumption of upper-bounded frequencies, we show a
randomized approximation algorithm that for each given $\beta$, $\beta
> 1$, and each given probability $1-\epsilon$, outputs in FPT time a
$\beta$-approximate solution with probability at least $1-\epsilon$
(the FPT time is with respect to $K$, $\beta$, and $\epsilon$).
Finally, in Section~\ref{sec:unrestricted} we consider two
exponential-time approximation algorithms for the unrestricted
MaxCover problem. Both of these algorithms solve a part of the problem
in a greedy way and a part using some exact algorithm, but they differ
in the order in which they apply each of these strategies. We show a
smooth transition between the running times of these algorithms and their
approximation ratios.

\subsection{Motivation}

We believe that the MaxCover problem with bounded frequencies is
an interesting and important problem on its own. However, the
particular reason why we study it is due to its connection to
winner-determination under Chamberlin--Courant's voting rule. Under
the Chamberlin--Courant's rule, a society of $n$ voters chooses from a
group of $m$ candidates a committee of $K$ representatives. The rule
was originally proposed as a mean of electing
parliaments~\cite{ccElection}, but recently Boutilier and
Lu~\cite{budgetSocialChoice} pointed out that it might be very useful
in the context of recommendation systems and Skowron et
al.~\cite{sko-fal-sli:c:multiwinner} showed its connection to 
resource allocation problems.

There are many variants of the rule, depending on the so-called
misrepresentation function that it uses. Here, we will focus on the
approval variant, though we mention that perhaps the best-studied one
(though not necessarily the most practical one) is the variant that
uses the Borda misrepresentation function (we omit the details of
Borda misrepresentation here and point the reader to the original
paper defining the rule~\cite{ccElection}).

In the approval-based variant of Chamberlin--Courant's rule, the
voters submit ballots on which they list all the candidates that they
find acceptable as their representatives (that is, the candidates that
they approve of). For each size-$K$ subset $S$ of candidates (referred
to as a \emph{committee}), the misrepresentation score of $S$ is the
total number of voters who do not approve of any of the candidates in
$S$. Chamberlin--Courant's rule elects a committee $S$ that minimizes
the misrepresentation.  Naturally, there may be several committees
that minimize the misrepresentation and in practice one has to apply
some tie-breaking. In the computational studies of voting researchers
are usually interested in finding any such a committee and so do we.

The above description makes clear the connection between
approval-based Chamberlin--Courant's rule and the MaxCover problem:
The voters are the elements that need to be covered, the candidates
are the sets (a voter $v$ belongs to the set defined by some candidate
$c$ if $v$ approves of $c$), and the size of the committee is the
number of sets one can pick. The achieved misrepresentation is the
number of uncovered elements.

Given this connection, clearly winner determination under
Chamberlin--Courant's rule is an $\np$-hard problem for the approval
misrepresentation~\cite{complexityProportionalRepr} (it also is for
Borda misrepresentation~\cite{budgetSocialChoice}). Further, in both
cases the problem is $\wtwo$-hard, as shown by Betzler et
al.~\cite{fullyProportionalRepr}.  Thus if one wants to find the exact
winning committee, one is restricted to exponential time algorithms,
such as, e.g., solving a particular integer linear
program~\cite{potthoff-brams} or trying all possible committees.  On
the other hand, at least for the Borda misrepresentation function, the
problem is quite easy to approximate, both theoretically (there is a
polynomial-time approximation scheme due to Skowron et
al.~\cite{sko-fal-sli:c:multiwinner}) and in practice (as shown by
experiments~\cite{sko-fal-sli:c:monroe-cc-experimental}).
Unfortunately, the connection between the approval variant of the rule
and the MaxCover problem severely limits approximation possibilities:
In terms of polynomial-time algorithms the best we can get is the
standard greedy $(1-\frac{1}{e})$-approximation algorithm.

Yet, in practical elections it is somewhat unreasonable to expect that
each voter would list many candidates as approved.  Indeed, in some
political systems that use approval-like ballots, even the law itself
limits the number of candidates one can list (for example, in Polish
parliamentary elections the voters can list up to three
candidates). Thus it is most natural to consider the approval variant
of Chamberlin--Courant's rule for the case where each voter can
approve of at most a given number $p$ of the candidates. This variant
of the rule directly corresponds to the MaxCover problem with bounded
frequencies. On the other hand, it is also natural to consider
settings were voters are required to approve of at least a given
number of candidates (such requirement can, for example, be imposed by
the election rules). This corresponds to the MaxCover problem were
elements' frequencies are \emph{lower bounded} by some value.

In effect, our results on the MaxCover problem with bounded
frequencies fill in the hole between efficient approximation
algorithms for the Borda variant of Chamberlin--Courant's rule given
by Skowron et al.~\cite{sko-fal-sli:c:multiwinner} and
$\wtwo$-hardness results of Betzler et
al.~\cite{fullyProportionalRepr} for the general, unrestricted
approval variant of the rule.

\section{Preliminaries}\label{sec:prelims}

We assume that the reader is familiar with standard notions regarding
(approximation) algorithms, computational complexity theory and
parametrized complexity theory. Below we provide a very brief review.
For each positive integer $n$, we write $[n]$ to mean $\{1, \ldots,
n\}$. 

Let $\calP$ be an algorithmic problem where, given some instance $I$,
the goal is to find a solution $s$ that maximizes a certain function
$f$.  Given an instance $I$ of $\calP$, we refer to the value $f(s)$
of an optimal solution $s$ as $\OPT(I)$ (or, sometimes, simply as
$\OPT$ if the instance $I$ is clear from the context). Let $\beta$, $0
< \beta \leq 1$, be some fixed constant.  An algorithm $\calA$ that
given instance $I$ returns a solution $s'$ such that $f(s') \geq
\beta\OPT(I)$ is called a $\beta$-approximation algorithm for the
problem $\calP$.
Analogously, we define $\OPT(I)$ and the notion of a
$\gamma$-approximation algorithm, $\gamma > 1$, for the case of a
problem $\calP'$, where the task is to find a solution that minimizes
a given goal function $g$. (Specifically, given an instance $I$ of
$\calP'$, a $\gamma$-approximation algorithm is required to return a
solution $s'$ such that $g(s') < \gamma \OPT(I)$).  Given instance $I$
of some algorithmic problem, we write $|I|$ to denote the length of
the standard, efficient encoding of $I$.

In this paper we focus on the following two problems.
\begin{definition}
  An instance $I = (N,\calS,K)$ of the MaxCover problem consists of a
  set $N$ of $n$ elements, a collection $\calS = \{S_1, \ldots, S_m\}$
  of $m$ subsets of $N$, and nonnegative integer $K$. The goal is to
  find a subcollection $\calC$ of $\calS$ of size at most $K$ that
  maximizes $\|\bigcup_{S \in \calC}S\|$.
\end{definition}

\begin{definition}
  The MinNonCovered problem is defined in the same way as the MaxCover
  problem, except the goal is to find a subcollection $\calC$ such
  that $\|N\| - \|\bigcup_{S \in \calC}S\|$ is minimal.
\end{definition}
In the decision variant of MaxCover (of MinNonCovered) we are
additionally given an integer $T$ (an integer $T'$) and we ask if
there is a collection of up to $K$ sets from $\calS$ that cover at
least $T$ elements (that leave at most $T'$ elements
uncovered). MaxVertexCover is a variant of MaxCover where we are given
a graph $G = (V,E)$, the edges are the elements to be covered, and
vertices define the sets that cover them (a vertex covers all the
incident edges). SetCover and VertexCover are variants of MaxCover and
MaxVertexCover, respectively, where we ask if it is possible to cover
all the elements (all the edges).

In terms of the optimal solutions, MaxCover and MinNonCovered are
equivalent. None\-theless, they do differ when considered from the
point of view of approximation. For example, if there were a solution
that covered all the $n$ elements, then a $\beta$-approximation
algorithm for MaxCover, $0 < \beta < 1$, would be free to return a
solution that covered only $\beta n$ of them, but a
$\gamma$-approximation algorithm for the MinNonCovered problem,
$\gamma > 1$, would have to provide an optimal solution that covered
all the elements.

Given an instance $I$ of MaxCover (MinNonCovered), we say that an
element $e$ has frequency $t$ if it appears in exactly $t$ sets.  We
mostly focus on the variants of MaxCover and MinNonCovered where there
is a given constant $p$ such that each element's frequency is at most
$p$. We refer to these problems as variants with bounded frequencies.

We will focus on (approximation) algorithms that run in FPT time (see
the books of Downey and Fellows~\cite{dow-fel:b:parameterized},
Niedermeier~\cite{nie:b:invitation-fpt}, and Flum and
Grohe~\cite{flu-gro:b:parameterized-complexity} for details on
parametrized complexity theory). To speak of an FPT algorithm for a
given problem, we declare a part of the problem as the so-called
parameter. Here, for MaxCover and MinNonCovered problems, we take the
parameter to be the number $K$ of sets that we are allowed to use in
the solution (in Section~\ref{sec:worst-case} we briefly consider
MaxCover/MinNonCovered with parameters $T$, $T'$, and their
combinations with $K$). Given an instance $I$ of a problem with
parameter $k$, an FPT algorithm is required to run in time
$f(k)\poly(|I|)$, where $f$ is some computable function and $\poly(\cdot)$ is
some polynomial.

From the point of view of parametrized complexity, FPT is seen as the
class of tractable problems. There is also a whole hierarchy of
hardness classes, $\fpt \subseteq \wone \subseteq \wtwo \subseteq
\cdots \wpclass \subseteq \cdots$. The standard definitions of $\wone,
\wtwo, \ldots$ are quite involved and so we point the reader to
appropriate
overviews~\cite{dow-fel:b:parameterized,nie:b:invitation-fpt,flu-gro:b:parameterized-complexity}.
However, we can also define these classes through an appropriate
reduction notion and their complete problems.

\begin{definition}
  Let $\mathcal{P}$ and $\mathcal{P}'$ be two decision problems
  parametrized by real non-negative parameters $\mathcal{K}$ and
  $\mathcal{K}'$, respectively. We say that $\mathcal{P}$ reduces to
  $\mathcal{P}'$ through a parametrized reduction if there exist a
  mapping $F \colon \mathcal{P} \rightarrow \mathcal{P}'$ (computable
  in FPT time with respect to parameter $\mathcal{K}$) and two
  computable functions, $g \colon \reals_{+} \rightarrow \reals_{+}$
  and $h: \reals_{+} \rightarrow \reals_{+}$, such that (i) for each
  instance $(I, K) \in \mathcal{P}$ the answer to $(I, K)$ is ``yes''
  if and only if the answer to $F(I) = (I', K')$ is ``yes'', (ii) $K$
  and $K'$ are the values of the parameters $\mathcal{K}$ and
  $\mathcal{K}'$ respectively, (iii) $|I'| \leq g(K)\poly(|I|)$,
  and (iv) $K' \leq h(K)$.
\end{definition}

$\wone$ is the class of all problems for which there is a parametrized
reduction to the Clique problem (i.e., the problem where we ask if a
given graph $G = (V,E)$ has a clique of size at least $K$, where $K$
is the parameter). $\wtwo$ is the class of problems with parametrized
reductions to SetCover (with parameter $K$).  Interestingly,
VertexCover is well-known to be in FPT, but MaxVertexCover is
$\wone$-complete~\cite{guo-nie-wer:j:vertex-cover-variants}.

One of the standard ways of showing $\wone$-membership is to give a
reduction to the
Short-Non\-de\-ter\-ministic-Turing-Machi\-ne-Com\-putation problem
(shown to be $\wone$-complete for parameter $k$ by
Cesati~\cite{ces:j:turing-way-parameterized-complexity}).

\begin{definition}
  In the Short-Non\-de\-ter\-ministic-Turing-Machi\-ne-Com\-putation
  problem we are given a single-tape nondeterministic Turing machine
  $M$ (described as a tuple including the input alphabet, the work
  alphabet, the set of states, the transiation function, the initial
  state and the accepting/rejecting states), a string $x$ over $M$'s
  input alphabet, and an integer $k$. The question is whether there is
  an accepting computation of $M$ that accepts $x$ within $k$ steps.
\end{definition}

The Bounded-Non\-de\-ter\-ministic-Turing-Machi\-ne-Com\-putation
problem is defined similarly, but in addition we are also given an
integer $m$, and we ask if $M$ accepts its input within $m$ steps, of
which at most $k$ are nondeterministic. Cesati has shown that this
problem is
$\wpclass$-complete~\cite{ces:j:turing-way-parameterized-complexity}
(we omit the exact definition of $\wpclass$; the reader can think of
$\wpclass$ as the set of problems that have parameterized reductions
to the Bounded-Non\-de\-ter\-ministic-Turing-Machi\-ne-Com\-putation
problem).

\section{Related Work}\label{sec:related}

There is extensive literature on the complexity and approximation
algorithms for the SetCover and VertexCover problems.
%
%
On the other hand, the literature on MaxCover and MaxVertexCover is
more scarce
The literature on MaxCover with bounded frequencies of the elements is
scarcer yet. Below we survey some of the known results.

First, it is immediate that MaxCover, MinNonCovered, and
MaxVertexCover are $\np$-complete (this follows immediately from the
$\np$-completeness of SetCover and VertexCover). In terms of
approximation, a greedy algorithm that iteratively picks sets that
cover the largest number of yet uncovered elements achieves the
approximation ratio $1-\frac{1}{e}$, and this is optimal unless $\p =
\np$ (see, e.g., the textbook~\cite{hoc:b:covers} for the analysis of
the greedy algorithm and the work of Feige~\cite{fei:j:cover} for the
approximation lower bound). However, our focus is on the MaxCover
problem with bounded frequencies and this problem is, in spirit,
closer to MaxVertexCover than to the general MaxCover problem. Indeed,
MaxVertexCover can be seen as a special case of MaxCover with
frequencies bounded by $2$. However, we stress that even MaxCover with
frequencies bounded by $2$ is considerably more general than
MaxVertexCover and, compared to MaxVertexCover, may require different
algorithmic insights.

As far as we know, the best polynomial-time approximation algorithm
for MaxVertexCover is due to Ageev and
Sviridenko~\cite{age-svi:b:covers}, and achieves approximation ratio
of $\frac{3}{4}$. However, in various settings, it is possible to
achieve better results; we mention the papers of Han et
al.~\cite{han-ye-zha-zha:j:cover} and of Galluccio and
Nobili~\cite{gal-nob:j:cover} as examples.

From the point of view of parametrized complexity, MaxVertexCover was
first considered by Guo et
al.~\cite{guo-nie-wer:j:vertex-cover-variants}, who have shown that it
is $\wone$-complete. The problem was also studied by
Cai~\cite{cai:j:cardinality-constrained} who gave the currently best
exact algorithm for it and by Marx, who gave an FPT approximation
scheme~\cite{Marx06parameterizedcomplexity}. There is also an FPT
algorithm for MaxCover , for parameter $T$, i.e., the number of
elements to cover, due to Bl{\"a}ser~\cite{bla:j:partial-set-cover}.

In our paper, we attempt to merge parametrized study of MaxCover with
its study from the point of view of approximation algorithms.  In that
respect, our work is very close in spirit that of Croce and
Paschos~\cite{cro-pas:j:cover}, who provide moderately exponential
approximation algorithms for MaxVertexCover, and to the work of
Marx~\cite{Marx06parameterizedcomplexity}.  Compared to their results,
we consider a more general problem, MaxCover (with or without bounded
frequencies) and, as far as it is possible, we seek algorithms that
run in FPT time (the algorithm of Croce and Paschos is not
FPT). Interestingly, even though we focus on a more general problem,
our algorithms improve upon the results of Croce and
Paschos~\cite{cro-pas:j:cover} and of
Marx~\cite{Marx06parameterizedcomplexity}, even when applied to
MaxVertexCover
.

\section{Worst-Case Complexity Results}\label{sec:worst-case}

We start our parametrized study of the MaxCover problem by
considering its worst-case complexity. We first consider MaxCover with
bounded frequencies. It follows directly from the literature that the
problem is $\wone$-hard, and here we show that it is, in fact,
$\wone$-complete (unless the frequency bound $p$ is exactly $1$; then
it is optimal to simply pick the sets with highest cardinalities).

\begin{theorem}
  For each constant $p$ greater than $2$, the MaxCover problem with
  frequencies upper-bounded by $p$ is $\wone$-complete (when
  parametrized by the number of sets in the solution).
\end{theorem}
\begin{proof}
  The hardness follows directly from the $\wone$-hardness of the
  MaxVertexCover problem~\cite{guo-nie-wer:j:vertex-cover-variants}.
  We prove membership in $\wone$ by reducing MaxCover with bounded
  frequencies to the Short-Nondeterministic-Turing-Machine-Computation
  problem.

  Let $p$ be some fixed constant and let $I = (N,\calS,K,L)$ be our
  input instance, where $N$ is a set of elements, $\calS = \{S_1,
  \ldots, S_m\}$ is a family of subsets of $N$ (each element from $N$
  appears in at most $p$ sets from $\calS$), and $K$ and $L$ are two
  integers. This is the decision variant of the problem, thus we have
  $L$ in the input; we ask if there is a collection of up to $K$ sets
  from $\calS$ that jointly cover at least $L$ elements.  W.l.o.g., we
  assume that $K \geq m$. We form single-tape nondeterministic Turing
  machine $M$ to execute the following algorithm (on empty input
  string); the idea of the algorithm is to employ the standard
  inclusion-exclusion principle:
  \begin{enumerate}
  \item Guess the indices $i_1, \ldots, i_K$ of $K$ sets from $\calS$.
  \item Set $T = 0$.
  \item For each subset $A$ of $\{i_1, \ldots, i_K\}$ of size up to
    $p$, do the following: If $\|A\|$ is odd, add $\|\bigcap_{i \in
      A}S_i\|$ to $T$, and otherwise subtract $\|\bigcap_{i \in
      A}S_i\|$ from $T$.
  \item If $T \geq L$ then we accept and otherwise we reject.
  \end{enumerate}
  It is easy to see that this algorithm can indeed be implemented on a
  single-tape nondeterministic Turing machine with a sufficiently
  large (but polynomially bounded) work alphabet and state space.  The
  only issue that might require a comment is the computation of
  $\|\bigcap_{i \in A}S_i\|$. Since sets $A$ contain at most $p$
  elements, we can precompute these values and store them in $M$'s
  transition function.

  The correctness of the algorithm follows directly from the inclusion-exclusion
  principle and the fact that each element appears in at most $p$ sets: 
%
  \begin{align*}
    \| S_{i_1} \cup S_{i_2} \cup \cdots \cup S_{i_K}\| &= \sum_{\ell \in [K]}\|S_{i_\ell}\| 
      -\!\!\! \sum_{\substack{\ell' \in [K] \\ \ell'' \in [K] \\ \ell' \neq \ell''}}\|S_{i_{\ell'}}\cap S_{i_{\ell''}}\| 
      &+ \sum_{\substack{\ell' \in [K]\\ \ell'' \in [K]\\ \ell''' \in [K] \\ \ell' \neq \ell'' \\ \ell' \neq \ell''' \\ \ell'' \neq \ell'''}}\|S_{i_{\ell'}}\cap S_{i_{\ell''}} \cap S_{i_{\ell'''}}\| 
      - \cdots 
  \end{align*}
  In general, the above formula should include intersections of up to
  $K$ sets. However, since in our case each element appears in at most
  $p$ sets, the intersection of more than $p$ sets are always
  empty. This shows that the algorithm is correct and concludes the
  proof.~
\end{proof}

For the sake of completeness, we mention that both the unrestricted
variant of the problem and the one where we put a lower bound on each
element's frequency are $\wtwo$-hard.

\begin{theorem}
  For each constant $p$, $p \geq 1$, MaxCover where each element
  belongs to at least $p$ sets if $\wtwo$-hard.
%
\end{theorem}
\begin{proof}
  To show $\wtwo$-hardness, we give a reduction from SetCover. In the SetCover
  problem we ask whether there exist $K$ subsets that cover all the elements
  (we give a reduction for the parameter $K$).
  Let $I = (N,\calS)$ be an input instance of
  SetCover.  W.l.o.g., we can assume that each element from $N$
  belongs to at least one set in $\calS$. We form an instance $I'$ of
  MaxCover which is identical to $I$, except (a) for each $e \in N$,
  we modify $\calS$ to  additionally include $p-1$ copies of set $\{e\}$, and
  (b) we run the MaxCover algorithm asking whether the maximal number of the
  elements covered by $K$ subsets is at least equal to $\|N\|$.
  Clearly, in $I'$ each element belongs to at least $p$ sets
  and $I'$ is a yes-instance of MaxCover if and only if $I$ is a
  yes-instance of SetCover.~
\end{proof}

So far, we were not able to show that MaxCover (even with
lower-bounded frequencies) is in $\wtwo$. Nonetheless, it is quite
easy to show that the problem belongs to $\wpclass$.

\begin{theorem}
  For each constant $p$, $p \geq 1$, MaxCover where each element
  belongs to at least $p$ sets is in $\wpclass$ (when parametrized by
  the number of sets in the solution).
\end{theorem}
\begin{proof}
  We give a reduction from MaxCover to the
  Bounded-Nondeterministic-Turing-Machine-Computation problem. On
  input $I = (N,\calS,K,T)$, where $N$, $\calS$, and $K$ are as usual
  and $T$ is the lower bound on the number of elements that we should
  cover, we produce a machine that on empty input executes the following
  algorithm:
  \begin{enumerate}
  \item It nondeterministically guesses up to $K$ names of sets from
    $\calS$ and writes these names on the tape (each name of a set
    from $\calS$ is a single symbol).
  \item Deterministically, for each name of the set produced in the
    previous step, the machine writes on the tape the names of those
    elements from this set that have not been written on the tape yet.
  \item The machine counts the number of names of elements written on
    the tape.  If there were at least $T$ of them, it accepts. Otherwise
    it rejects.
  \end{enumerate}
  It is easy to see that we can produce a description of such a
  machine in polynomial time with respect to $|I|$. Further, it is
  clear that its nondeterministic running time is bounded by some
  polynomial of $|I|$ and that it makes at most $k$ nondeterministic
  steps.
\end{proof}

It is quite interesting to also consider MaxCover with other
parameters.  First, recall that for parameter $T$, the number of
elements that we should cover, Bl{\"a}ser has shown that MaxCover is
in FPT~\cite{bla:j:partial-set-cover}.  What can we say about
parameter $T' = n-T$, i.e., the number of elements we can leave
uncovered (this, in essence, means considering the MinNonCovered
problem, but for the worst-case setting it is more convenient to speak
of the parameter $T'$)? In this case, the problem is immediately seen
to be para-$\np$-complete (that is, the problem is $\np$-complete even
for a constant value of the parameter).

\begin{corollary}\label{cor:teqzero}
  The MaxCover problem is para-$\np$-complete when parametrized by the
  number $T'$ of elements that can be left uncovered. This holds even
  if each element's frequency is upper-bounded by some constant $p$,
  $p \geq 2$.
\end{corollary}
\begin{proof}
  The following trivial reduction from SetCover suffices: Given an
  input instance $I = (N,\calS,K)$, output an instance
  $(N,\calS,K,0)$, i.e., an identical one, where we require that the
  number of elements left uncovered is $0$. Since the reduction is
  clearly correct and works for the constant value of the parameter,
  we get pare-$\np$-completeness. To obtain the result for upper-bounded frequencies,
  simply use VertexCover instead of SetCover in the reduction.
\end{proof}

However, if we consider the joint parameter $(K,T')$, then the
MaxCover problem becomes $\wtwo$-complete.

\begin{theorem}
  MaxCover is $\wtwo$-complete when parametrized by both the number
  $K$ of sets that can be used in the solution and the number $T'$ of
  elements that can be left uncovered.
\end{theorem}
\begin{proof}
  We obtain $\wtwo$-hardness by simply observing that the reduction
  given in Corollary~\ref{cor:teqzero} suffices. To prove
  $\wtwo$-membership, we give a reduction from MaxCover (with
  parameter $(K,T')$) to SetCover (with parameter $K$).

  Let $I = (N,\calS,K,T')$ be an input instance of MaxCover.  We
  form an instance $I' = (N',\calS',K+T')$ of SetCover as follows.
  Let $N' = N \cup D' \cup D''$, where $D' = \{d'_1, \ldots, d'_K\}$
  and $D'' = \{d''_1, \ldots, d''_{T'}\}$. For each set $S \in \calS$
  and each $d_i' \in D'$, we set $S(d_i') = S \cup \{d_i'\}$.  We set
  $\calS' = \calS'_1 \cup \calS'_2$, where (a) $\calS'_1 = \{ S(d_i'):
  (S \in \calS) \land (d_i' \in D') \}$, and (b) $\calS'_2 = \{ \{e,d''_i\}:
  e \in N, d''_i \in D''\}$.

  It is easy to see that if $I$ is a yes-instance of MaxCover
  then $I'$ is a yes-instance of SetCover: If for $I$ it is possible
  to cover $n-T'$ elements of $N$ using $K$ sets, then for $I$ it is
  possible to (a) use $K$ sets from $\calS'_1$ to cover $n-T'$
  elements from $N$ and all the elements from $D'$, and (b) use $T'$
  sets from $\calS'_2$ to cover all the elements from $D''$ and the
  remaining $T'$ elements from $N$.  For the other direction, assume
  that $I'$ is a yes-instance of SetCover. However, covering the
  elements from $D'$ requires one to use at least $K$ sets from
  $\calS'_1$ (which correspond to the sets from $\calS$) and
  covering the elements in $D''$ requires at least $T'$ sets from
  $\calS'_2$. Since each set from $\calS'_2$ covers exactly one
  element from $N$, it is easy to see that if $I'$ is a yes-instance,
  then it must be possible to cover at least $\|N\|-T'$ elements from
  $N$ using $K$ sets from $\calS$.
\end{proof}


\begin{table}
  \begin{center}
    \begin{tabular}{r|l}
      parameter  & {worst-case complexity of MaxCover} \\
      \hline
      \multirow{2}{*}{$K$}      & {$\wtwo$-hard, in $\wpclass$}\\
               &  $\wone$-complete for upper-bounded frequencies\\
      \rule{0cm}{5.5mm} 
      $T$      & FPT~\cite{bla:j:partial-set-cover} \\
      $(K,T)$  & FPT~\cite{bla:j:partial-set-cover} \\
      \rule{0cm}{5.5mm} 
      $T'$     & para-$\np$-complete \\
      $(K,T')$ & $\wtwo$-complete   \\
  \end{tabular}
  \caption{\label{tab:complexity}Parameterized worst-case complexity
    results for unrestricted MaxCover and MinNonCovered. The
    parameters are as follows: $K$ is the number of sets we can use in
    the solution, $T$ is the number of elements we are required to
    cover, and $T' = n-T$ is the number of elements we can leave
    uncovered.}
  \end{center}
\end{table}

We summarize our worst-case complexity resutls in
Table~\ref{tab:complexity}. Not surprisingly, using the parameter $T'$
(i.e., in essence, considering the MinNonCovered problem) leads to
higher computational complexity than using parameter $T$ (i.e., in
essence, considering the MaxCover problem). For the parameter $K$, the
exact complexity of unrestricted MaxCover remains open.

\section{Algorithms for the Case of Bounded Frequencies}\label{sec:approximation}

In this section we present our approximation algorithms for the
MaxCover and MinNonCovered problems, for the case where we either
upper-bound or lower-bound the frequencies of the elements. We first
consider the with MaxCover problem, both with upper-bounded
frequencies and with lower-bounded frequencies, and then move on to
the MinNonCovered problem with upper-bounded frequencies.

\subsection{The MaxCover Problem with Upper Bounded Frequencies}

We will now present an FPT approximation scheme for MaxCover with
upper-bounded frequencies. While
Marx~\cite{Marx06parameterizedcomplexity} has already shown an FPT
approximation scheme for MaxVertexCover, his approach cannot be
directly generalized to the MaxCover problem with bounded frequencies
(although there are some similarities between the algorithms).
Also, our algorithm for MaxCover applied to the MaxVertexCover problem
is considerably faster than the algorithm of Marx~\cite{Marx06parameterizedcomplexity}.
We will give a brief comparison of the two algorithms after presenting our approach.

Intuitively, our algorithm works in a very simple way. Given an
instance $I = (N,\calS,K)$ of MaxCover (with frequences bounded by
some constant $p$) and a required approximation ratio $\beta$, the
algorithm simply picks some of the sets from $\calS$ with highest
cardinalities (the exact number of these sets depends only on $K$, $p$, and
$\beta$), tries all $K$-element subcollections of sets from this
group, and returns the best one. This approach is formalized as
Algorithm~\ref{alg:pApproval}. The following theorem explains that
indeed the algorithm achieves a required approximation ratio.

\SetKwInput{KwParameters}{Parameters}
\begin{algorithm}[t]
   \small
 \KwParameters{\\$\hspace{3pt}$ $(N,\calS,K)$ --- input MaxCover instance\\
          $\hspace{3pt}$ $p$ --- bound on the number of sets each element can belong to\\
          $\hspace{3pt}$ $\beta$ --- the required approximation ratio of the algorithm \\}
   \vspace{2mm}
   \SetAlCapFnt{\small}
   $\mathcal{A} \leftarrow \lceil \frac{2p K}{(1 - \beta)} + K \rceil$ sets from $\calS$ with the highest cardinalities \;
   \ForEach{$K$-element subset $\mathcal{C}$ of $\mathcal{A}$}{
       $\quality[\mathcal{C}] \leftarrow$ the number of elements covered by $\mathcal{C}$\;
   }
   \Return{$\argmax_{\mathcal{C}} (\quality[\mathcal{C}])$} \;
   \caption{\small The algorithm for the MaxCover problem with
     frequency upper bounded by $p$.}
   \label{alg:pApproval}
\end{algorithm}

\begin{theorem}\label{thm:pApproval}
  For each instance $I = (N,\calS,K)$ of MaxCover where each element
  from $N$ appears in at most $p$ sets in $\calS$,
  Algorithm~\ref{alg:pApproval} outputs a $\beta$-approximate solution
  in time $\poly(n,m) \cdot {\frac{2p K}{(1 - \beta)} + K \choose K}$.
\end{theorem}
\begin{proof}
  It is immediate to establish the running time of the algorithm.  We
  show that its approximation ratio is, indeed,~$\beta$.

  Consider some input instance $I$.  Let $\calC$ be the solution
  returned by Algorithm~\ref{alg:pApproval} and let $\calC^{*}$ be
  some optimal solution. Let $c$ be an arbitrary function such that
  for each element $e$ such that $\exists_{S \in \calC^{*}}: e \in S$,
  $c(e)$ is some $S \in \calC^{*}$ such that $e \in S$. We refer to
  $c$ as the \emph{coverage function}. Intuitively, the coverage
  function assigns to each element covered under $\calC^{*}$ (by,
  possibly, many different sets) the particular set ``responsible''
  for covering it.
%
%
  We say that $S$ covers $e$ if and only if $c(e) = S$. Let $\OPT$
  denote the number of elements covered by $\calC^{*}$.

  We will show that $\calC$ covers at least $\beta\OPT$
  elements. Naturally, the reason why $\calC$ might cover fewer
  elements than $\calC^{*}$ is that some sets from $\calC^*$ may not
  be present in $\calA$, the set of the subsets considered by the algorithm.
  We will show an iterative procedure that
  starts with $\calC^{*}$ and, step by step, replaces those members of
  $\calC^*$ that are not present in $\calA$ with the sets from
  $\calA$. The idea of the proof is to show that each such replacement
  decreases the number of covered element by at most a small amount.

  Let $\ell = \|\calC^{*} \setminus \calC\|$. Our procedure will
  replace the $\ell$ sets from $\calC^*$ that do not appear in $\calC$
  with $\ell$ sets from $\calA$.  We renumber the sets so that
  $\calC^{*} \setminus \calC = \{S_1, \ldots, S_\ell\}$. We will
  replace the sets $\{S_1, \ldots, S_\ell\}$ with sets $\{S'_1,
  \ldots, S'_\ell\}$ defined through the following algorithm.  Assume
  that we have already computed sets $S'_1, \ldots, S'_{i-1}$ (thus
  for $i=1$ we have not yet computed anything). We take $S'_i$ to be a
  set from $\mathcal{A} \setminus (\mathcal{C}^{*} \cup \{S_1', \dots,
  S_{i-1}'\})$ such that the set $(\mathcal{C}^{*} \setminus \{S_1,
  \dots, S_{i}\}) \cup \{S_1', \dots, S_{i}'\}$ covers as many
  elements as possible. During the $i$'th step of this algorithm,
  after we replace $S_i$ with $S_i'$ in the set $(\mathcal{C}^{*}
  \setminus \{S_1, \dots, S_{i-1}\}) \cup \{S_1', \dots, S_{i-1}'\}$,
  we modify the coverage function as follows:
  \begin{enumerate}
  \item for each element $e$ such that $c(e) = S_i$, we set $c(e)$ to be undefined;
  \item for each element $e \in S'_i$, if $c(e)$ is undefined then we set $c(e) = S'_i$.
  \end{enumerate}

  After replacing $S_i$ with $S'_i$, it may be the case that fewer
  elements are covered by the resulting collection of sets. Let $x_i$
  denote the difference between the number of elements covered by
  $(\mathcal{C}^{*} \setminus \{S_1, \dots, S_{i}\}) \cup \{S_1',
  \dots, S_{i}'\}$ and by $(\mathcal{C}^{*} \setminus \{S_1, \dots,
  S_{i-1}\}) \cup \{S_1', \dots, S_{i-1}'\}$ (or $0$, if by a
  fortunate coincidence there are more elements covered after
  replacing $S_i$ with $S'_i$). By the construction of the set $\calA$
  and the fact that $S_i \notin \calA$, each set from $\mathcal{A}$
  contains more elements than $S_i$. Thus we infer that every set from
  $\mathcal{A} \setminus (\mathcal{C}^{*} \cup \{S_1', \dots,
  S_{i-1}'\})$ must contain at least $x_i$ elements covered by
  $(\mathcal{C}^{*} \setminus \{S_1, \dots, S_{i-1}\}) \cup \{S_1',
  \dots, S_{i-1}'\}$. Indeed, if some set $S' \in \mathcal{A}
  \setminus (\mathcal{C}^{*} \cup \{S_1', \dots, S_{i-1}'\})$
  contained fewer than $x_i$ elements covered by $(\mathcal{C}^{*}
  \setminus \{S_1, \dots, S_{i-1}\}) \cup \{S_1', \dots, S_{i-1}'\}$,
  $S'$ would have to cover at least \[\|S'\|-(x_i-1) \geq \|S_i\|
  -(x_i-1)\] elements uncovered by $(\mathcal{C}^{*} \setminus \{S_1,
  \dots, S_{i-1}\}) \cup \{S_1', \dots, S_{i-1}'\}$. But this would
  mean that after replacing $S_i$ with $S'$, the difference between
  the number of covered elements would be at most $(x_i - 1)$.


  Let $\calC^{*}_{2}$ denote the set obtained after the
  above-described $\ell$ iterations. Since, for each $i$, the set
  $(\calC^{*} \setminus \{S_1, \dots, S_{i-1}\}) \cup \{S_1', \dots,
  S_{i-1}'\}$ is a subset of $\calC^{*} \cup \calC^{*}_{2}$, we know
  that, for each $i$, each set from $\mathcal{A} \setminus
  (\mathcal{C}^{*} \cup \{S_1', \dots, S_{\ell}'\})$ (there is
  $\|\mathcal{A}\| - K$ such sets) must contain at least $x_i$
  elements covered by $\calC^{*} \cup \calC^{*}_{2}$ (there is at most
  $2\OPT$ such elements).  Since each element is contained in at most
  $p$ sets, we infer that for each $i$, $x_i(\|\mathcal{A}\| - K) \leq
  2\OPT p$ and, as a consequence, $x_i \leq \frac{2\OPT
    p}{\|\mathcal{A}\| - K} = \frac{2\OPT p(1-\beta)}{2pK}$.  Thus we
  conclude that (recall that $\ell \leq K$):
  \begin{align*}
    \sum_{i=1}^\ell x_i \leq 2\OPT p K \frac{(1 - \beta)}{2p K} = (1 - \beta)\OPT
  \end{align*}
  That is, after our process of replacing the sets from $\calC^*$ that
  do not appear in $\calC$ with sets from $\calA$, at most
  $(1-\beta)\OPT$ elements fewer are covered. This means that there
  are $K$ sets in $\calA$ that together cover at least $\beta\OPT$
  elements. Since the algorithm tries all size-$K$ subsets of $\calA$,
  it finds a solution that covers at least $\beta\OPT$ elements.~
\end{proof}

Our analysis is tight up to the constant factor of $\frac{3}{4}$.
Below we present a family of parameters $\beta$ and instances of
MaxCover with upper-bounded frequencies on which our algorithm
achieves approximation ratio $(\frac{3}{4} + \frac{3}{4}\beta)$


\begin{proposition}
  There is a family $\calI$ of pairs $(I,\beta)$ where $I$ is an
  instance of MaxCover with bounded frequencies and $\beta$ is a real
  number, $0 < \beta < 1$, such that for each $(I,\beta) \in \calI$,
  if we use Algorithm~\ref{alg:pApproval} to find a
  $\beta$-approximate solution for $I$, it outputs an at-most
  $((\frac{3}{4} + \frac{3}{4}\beta)\OPT(I))$-approximate one.
\end{proposition}
\begin{proof}
  We describe how to construct pairs $(I,\beta)$ from the set
  $\calI$. We let $p$ be the bound of the frequencies of elements in
  $I$ and we let $K$ be the number of sets that we can use in the solution. We choose
  $p$ and $K$ to be sufficiently large, and $\beta$ to be sufficiently
  close to $1$ (the exact meaning of ``sufficiently large'' and
  ``sufficiently close to $1$'' will become clear at the end of the
  proof; elements of $\calI$ differ in the particular choices of $p$,
  $K$, and $\beta$).  We require that $\frac{1}{1-\beta}$ is an
  integer and that $p$ divides $K$.

  We now proceed with the construction of instance $I = (N,\calS,K)$
  for our choice of $p$, $K$, and $\beta$.  We set $x = \frac{2p K}{(1
    - \beta)} + K$; $x$ is the number of highest-cardinality sets from
  $\calS$ that Algorithm~\ref{alg:pApproval} will consider on instance
  $I$. By our choice of $\beta$ and $K$, $x$ is an integer and is
  divisible by $p$.  We form $N$, the set of elements to be covered,
  to consist of two disjoint subsets, $N_1$ and $N_2$, such that
  $\|N_1\| = {x \choose p}$ and $\|N_2\| = {x \choose p}
  \frac{Kp}{x}$.  We form the family $\calS$ to consist of two
  subfamilies, $\calS_1$ and $\calS_2$, defined as follows:
  \begin{enumerate}
  \item There are $x$ subsets in $\calS_1$, $\calS_1 = \{S_1, \ldots,
    S_x\}$. We form the sets in $S_1$ so that: (a) sets from $\calS_1$
    are subsets of $N_1$, (b) each element from $N_1$ belongs to
    exactly $p$ different sets from $\calS_1$, and (c) no two elements
    from $N_1$ belong to the same $p$ sets from $\calS_1$.
    Specifically, we build sets $(S_1, \ldots, S_m)$ as follows. Let
    $f$ be some one-to-one mapping between elements in $N_1$ and
    $p$-element subsets of $[x]$. For each $e \in N_1$, $e$ belongs
    exactly to the sets $S_{i_1}, \ldots, S_{i_p}$ such that $f(e) =
    \{i_1, \ldots, i_p\}$.  Note that each set $S_i \in \calS_1$
    contains exactly ${x-1 \choose p-1} = {x \choose p}\frac{p}{x}$
    elements.

  \item $\calS_2$ contains $K$ sets, each covering exactly ${x \choose
      p} \frac{p}{x}$ different elements from $N_2$ (and no other
    elements) so that no two sets from $\calS_2$ overlap.
  \end{enumerate}
  This completes our description of $I$. It is easy to see that each
  optimal solution for $I$ covers exactly $K{x \choose p}\frac{p}{x}$
  elements; each set contains exactly ${x \choose p}\frac{p}{x}$
  elements and, there are $K$ that are pairwise disjoint (for example
  the $K$ sets in $\calS_2$).

  Nonetheless, Algorithm~\ref{alg:pApproval} is free to choose any $x$
  sets from $\calS$ to include within $\calA$, the collection of sets
  from which it forms the solution, and, in particular, it is free to
  pick the $x$ sets from $\calS_1$.\footnote{We could also ensure that
    each set in $\calS_1$ contained one of $\frac{x}{p}$ additional
    elements, forcing the algorithm to pick exactly the sets from
    $\calS_1$, but that would obscure the presentation of our
    argument.} 

  Let us fix some arbitrary collection $\calS'$ of $K$ sets from
  $\calS_1$. For each $j$, $0 \leq j \leq K$, let $h(j)$ be the number
  of elements from $N_1$ that belong to exactly $j$ sets in $\calS'$.
  The number of elements covered by $\calS'$ is exactly $K{x \choose
    p}\frac{p}{x} - \sum_{j=2}^K(j-1)h(j)$.  How to compute $h(j)$?
  Using mapping $f$, it suffices to count the number of $p$-element
  subsets of $[x]$ that contain the indices of exactly $j$ sets from
  $\calS'$. In effect, we have $h(j) = {K \choose j}{{x - K} \choose
    p-j}$. We upper bound the number of sets covered by $\calS'$ with:
  \[ 
  K{x \choose p}\frac{p}{x} - h(2) = K{x \choose p} \frac{p}{x} - {K \choose 2}{x-K \choose p-2}.
  \]


  Consequently, on instance $I$ Algorithm~\ref{alg:pApproval} achieves
  the following approximation ratio $\frac{K{x \choose p} \frac{p}{x}
    - {K \choose 2}{x-K \choose p-2}}{K{x \choose p} \frac{p}{x}}$,
  which is equal to:
  \begin{align*}
    1 - \frac{{K \choose 2}{x-K \choose p-2}}{K{x \choose p}
      \frac{p}{x}}
  =  1 - \frac{{K \choose 2}{x -K \choose p}
      \frac{p(p-1)}{(x-K-p+2)(x-K-p+1)}}{K{x \choose p} \frac{p}{x}}
    \textrm{.}
  \end{align*}
  Now, if $x$ is large in comparison with $p$ and $K$ (which happens
  for sufficiently large $\beta$), then $\frac{{x -K \choose p}}{{x
      \choose p}} \approx 1$. Also, for sufficiently large $x$ and $p$
  (and for $x \gg p, K$) we have $\frac{p}{x-K-p+2} \approx
  \frac{p}{x}$ and $\frac{p-1}{x-K-p+1} \approx \frac{p}{x}$. Finally,
  for sufficiently large $K$ we have ${K \choose 2} \approx
  \frac{K^2}{2}$. Thus, for large values of $\beta$, $K$, and $p$, we
  can approximate the above ratio with the following expression:
  \begin{align*}
    1 - \frac{\frac{K^2}{2} \cdot \frac{p^2}{x^2}}{K\frac{p}{x}} = 1 - \frac{1}{2} \cdot \frac{Kp}{\frac{2p K}{(1 - \beta)} + K} \approx 
    1 - \frac{1}{2} \cdot \frac{Kp}{\frac{2p K}{(1 - \beta)}} = 1 -
    \frac{1}{4} \cdot(1 - \beta) = \frac{3}{4} + \frac{3}{4}\beta
    \textrm{.}
  \end{align*}
%
%
  This completes our argument.~
\end{proof}

Let us now compare our algorithm to that of
Marx~\cite{Marx06parameterizedcomplexity} for the case of MaxVertexCover.
Briefly put, the idea behind Marx's algorithm is as follows: Consider
vertices in the order of nonincreasing degrees. If the degree of the
vertex with the 
highest degree is large enough, then $K$ vertices with
the highest degrees already cover sufficiently many edges to give a
desired approximate solution. If the 
highest degree is not large enough, then there is an exact,
color-coding based, FPT algorithm that solves the problem
optimally. Our algorithm is similar in the sense that we also focus on
a group of sets with highest cardinalities (sets' cardinalities in
MaxCover correspond to vertex degrees in MaxVertexCover). However,
instead of simply picking $K$ largest ones, we make a careful decision
as to which exactly to take.\footnote{Indeed, it is possible to build
  an example where picking sets with highest cardinalities would not
  work. This trick works in Marx's algorithm because he considers
  graphs and, thus, can bound the negative effect of covering the same
  element by different sets; in the MaxCover problem this seems
  difficult to do.}
%
%
Further, our algorithm has a better running time than that of Marx.
To achieve approximation ratio $\beta$, the algorithm presented by
Marx has running time at least
$\Omega((\frac{k^3}{1-\beta})^{(\frac{k^3}{1-\beta})})$. For us, the
exponential factor in the running time is ${\frac{2p K}{(1 - \beta)} +
  K \choose K}$. On the other hand, we should point out that Marx's
algorithm's running time stems mostly from the exact part and the
algorithm given there is interesting in its own right.

\subsection{The MaxCover Problem with Lower-Bounded Frequencies}

Let us now move on the case of MaxCover with lower-bounded
frequencies.  It turns out that in this case the standard greedy
algorithm, given here as Algorithm~\ref{alg:greedy}, can---for
appropriate inputs---achieve a better approximation ratio than in the
unrestricted case.

\begin{algorithm}[t]
   \small
   \SetAlCapFnt{\small}
   \KwParameters{\\$\hspace{3pt}$ $(N,\calS,K)$ --- input MaxCover instance\\
      $\hspace{3pt}$ $p$ --- lower bound on the number of the sets each element belongs to}
   \vspace{3mm}

   $C = \{\}$\;
   \For{$i\leftarrow 1$ \KwTo $K$}{
      $\Cov \leftarrow \{e \in N: \exists_{S \in C} e \in S\}$ \;
      $S_{\best} \leftarrow \argmax_{S \in \{S_1, \dots, S_m\} \setminus C}$ 
      $\{e \in N \setminus \Cov: e \in S\}\|$\;
      $C \leftarrow C \cup \{S_{\best}\}$
   }
   \Return{C}
   \caption{\small The algorithm  for the MaxCover problem with frequency lower bounded by $p$.}
   \label{alg:greedy}
\end{algorithm}

\begin{theorem}\label{theorem:greedy}
  Algorithm~\ref{alg:greedy} is a polynomial-time $(1 -
  e^{-\frac{pK}{m}})$-approximation algorithm 
  for the MaxCover problem with frequency lower bounded by $p$,
  on instances with $m$ elements where we can pick up to $K$ sets.
\end{theorem}
\begin{proof}
  The algorithm clearly runs in polynomial time and so we show it's
  approximation ratio. Let $I = (N,\calS,K)$ be an input instance of
  MaxCover and let $p$ be an integer such that each element from $N$
  belongs to at least $p$ sets from $\calS$.
  
  We prove by induction that for each $i$, $0 \leq i \leq K$, after
  the $i$'th iteration of Algorithm~\ref{alg:greedy}'s main loop, the
  number of uncovered elements is at most $n(1 -
  \frac{p}{m})^{i}$. Naturally, for $i=0$ the number of uncovered
  elements is exactly $n$, the total number of elements.  Suppose that
  the inductive assumption holds for some $(i-1)$, $1 \leq i < K$ and
  let $x$ be the number of elements still uncovered after the
  $(i-1)$-th iteration (by the inductive assumtpion, we have $x \leq
  n(1 - \frac{p}{m})^{i-1}$). Since each element belongs to at least
  $p$ sets and neither of the sets containing the uncovered elements
  is yet selected, by the pigeonhole principle there is a
  not-yet-selected set that contains at least $\lceil x\frac{p}{m}
  \rceil$ of the uncovered elements.  In consequence, the number of
  elements still uncovered after the $i$-th iteration is at most:
  \begin{align*}
    x - x\frac{p}{m} = x\left(1 - \frac{p}{m}\right) \leq n\left(1 - \frac{p}{m}\right)^{i}
    \textrm{.}
  \end{align*}
  Thus after $K$ iterations the number of uncovered elements is at
  most:
  \begin{align*}
    n\left(1 - \frac{p}{m}\right)^{K} = n\left(1 - \frac{p}{m}\right)^{\frac{m}{p}\cdot
      \frac{pK}{m}} \leq ne^{-\frac{pK}{m}}\textrm{.}
  \end{align*}
  Since the number of covered elements in the optimal solution is at
  most $n$, the algorithm's approximation ratio is $(1 -
  e^{-\frac{pK}{m}})$.~
\end{proof}

Naturally, the standard approximation ratio of $(1 - e^{-1})$ of the
greedy algorithm still applies and we get the following corollary.

\begin{corollary}
Algorithm~\ref{alg:greedy} gives approximation guarantee of $(1 - e^{-\max(\frac{pK}{m}, 1)})$.
\end{corollary}

The analysis given in Theorem~\ref{theorem:greedy} is tight. Below we
present a family of instances on which the algorithm reaches exactly
the promised approximation ratio.

\begin{proposition}
  For each $\alpha$, $\alpha \geq 1$, there is an instance $I(\alpha)$
  of MaxCover (with $m$ sets. element frequency lower-bounded by $p$,
  $K$ sets to use, and $\frac{pK}{m} = \alpha$) such that on input
  $I(\alpha)$, Algorithm~\ref{alg:greedy} achieves approximation ratio
  no better than $(1 - e^{-\frac{pK}{m}})$.
\end{proposition}
\begin{proof}
  Let us fix some $\alpha$, $\alpha > 1$.  We choose integers $p$,
  $K$, and $m$ so that: (a) $p = \frac{\alpha m}{K}$, (b) $m \gg K$
  (and, thus, $p \gg K$), and (c) $p$, $m$, and $K$ are sufficiently
  large (the exact meaning of ``sufficiently large'' will become clear
  at the end of the proof). 

  We form instance $I(\alpha) = (N,\calS,K)$ as follows.
  %
  %
  We let $N = N_1 \cup \cdots \cup N_K$, where $N_1, \ldots, N_K$ are
  pairwise-disjoint sets, each of cardinality ${m-K \choose p-1}$
  (thus $\|N\| = K{m-K \choose p-1}$).  The family $\calS$ consists of
  two subfamilies, $\calS_1$ and $\calS_2$:
  \begin{enumerate}
  \item $\calS_1$ consists of $m - K$ sets, $S_{1}, \ldots, S_{m-K}$,
    constructed as follows. For each $i$, $1 \leq i \leq K$, let $f_i$
    be some one-to-one mapping from $N_i$ to $(p-1)$-element subsets
    of $[m-K]$. For each $i$, $1 \leq i \leq K$, if $e \in N_i$ and
    $f_i(e) = \{j_1, \ldots, j_{p-1}\}$ then we include $e$ in sets
    $S_{j_1}, S_{j_2}, \ldots, S_{j_{p-1}}$.  Note that for each
    $S_\ell$ in $\calS_2$, $\|S_\ell\| = K { m - K \choose p-2 }$; for
    each $i$, $1 \leq i \leq K$, $S_\ell$ contains ${m - K \choose
      p-2}$ elements from $N_i$; to see this, it suffices to count how
    many $(p-1)$-elements subsets of $[m-K]$ there are that contain $j$.

  \item $\calS_2 = \{N_1, \ldots, N_K\}$.
  \end{enumerate}
  Note that, by our construction, each element from $N$ belongs to
  exactly $p$ sets from $\calS$ ($p-1$ from $\calS_1$ and one from
  $\calS_2$).

  Naturally, the $K$ disjoint sets from $\calS_2$ form the optimal
  solution and cover all the elements.  We will now analyze the
  operation of Algorithm~\ref{alg:greedy} on input $I(\alpha)$.

  We claim that Algorithm~\ref{alg:greedy} will select sets from
  $\calS_1$ only. We show this by induction. Fix some $\ell$, $1 \leq
  \ell \leq K$, and suppose that until the beginning of the $\ell$'th
  iteration the algorithm chose sets from $\calS_1$ only.
  This means that, for each $i$, $1 \leq i \leq K$, each set $N_i$
  contains exactly ${m-K-\ell \choose p-1}$ uncovered elements. Why is
  this the case? Assume that the algorithm selected sets $S_{j_1},
  \ldots, S_{j_\ell}$. An element $e \in N_i$ is uncovered if and only
  if $f_i(e) \cap \{j_1, \ldots, j_\ell\} = \emptyset$; ${m-K-\ell
    \choose p-1}$ is the number of $(p-1)$-element subsets of $[m-K]$ that
  do not contain any members of $\{j_1, \ldots, j_\ell\}$.  So, if in
  the $\ell$'th iteration the algorithm choses some set from $\calS_2$,
  it would cover these additional ${m-K-\ell \choose p-1}$
  elements. On the other hand, if it chose a set from $\calS_1$, it
  would additionally cover $Kx$ elements, where $x = {m-K-\ell \choose
    p-1} - {m-K-\ell-1 \choose p-1}$.  By our choice, we have $pK > m$
  and, thus, $K > \frac{m-K}{p-1}$. We can now see that the
  following holds:
  \begin{align*}
    Kx &= K\left({m-K-\ell \choose p-1} - {m-K-\ell-1 \choose p-1}\right) 
     = K{m-K-\ell-1 \choose p-2}\\ &= \frac{K(p-1)}{m-K-\ell}{m-K-\ell \choose p-1} 
     \geq K\frac{p-1}{m-K}{m-K-\ell \choose p-1} > {m-K-\ell
      \choose p-1} \textrm{.}
  \end{align*}
  That is, in the $\ell$'th iteration Algorithm~\ref{alg:greedy} picks
  a set from $\calS_1$. This proves our claim.

  Let us now assess the approximation ratio Algorithm~\ref{alg:greedy}
  achieves on $I(\alpha)$. By the above reasoning, we know that it
  leaves ${m-2K \choose p-1}$ uncovered elements in each $N_i$, $1
  \leq i \leq K$.  Thus the fraction of the uncovered elements is
  bounded by the following expression (see some explanation below):
  \begin{align*}
    \frac{K{m-2K \choose p-1}}{K{m - K \choose p-1}} &= \frac{(m-2K)!(m-p-K+1)!}{(m-K)!(m-p-2K+1)!} \\ &= 
    \frac{(m-p-K+1)(m-p-K)\dots(m-p-2K)}{(m-K)(m-K-1)\dots(m-2K+1)} \\ & \geq 
    \left(\frac{m-2K-p}{m-2K+1}\right)^{K} = \left(1 -
      \frac{p+1}{m-2K+1}\right)^{K} \approx e^{-\frac{pK}{m}}
    \textrm{.}
  \end{align*}
  The first inequality holds by iterative application of the simple
  observation that if $1 \leq x \leq y$ then $\frac{x-1}{y-1} \leq
  \frac{x}{y}$.  To obtain the final estimate, we observe that for
  sufficiently large $p$ and $m$ (where $m \gg K$), we have
  $\frac{p+1}{m-2K+1} \approx \frac{p}{m} = \frac{\alpha}{K}$. For
  sufficiently large $K$, $(1 - \frac{\alpha}{K})^K \approx
  e^{-\alpha} = e^{-\frac{pK}{m}}$ (by the fact that $p = \frac{\alpha
    m}{K}$).  Since the optiomal solution covers all the elements, we
  have that Algorithm~\ref{alg:greedy} on input $I(\alpha)$ achieves
  approximation ratio no better than $1-e^{-\frac{pK}{m}}$.~
\end{proof}

Theorem~\ref{theorem:greedy} has some interesting implications. Let us
consider a version of the MaxCover problem in which the ratio
$\frac{p}{m}$ between the frequency lower bound $p$ and the number of
sets $m$ is constant.  This problems arises, e.g., if we use
approval-based variant of the Chamberlin-Courant's election system
with a requirement that each voter must approve at least some constant
fraction (e.g. 10\%) of the candidates. There exists a polynomial-time
approximation scheme (PTAS) for this version of the problem.

\begin{definition}
  For each $\alpha$, $0 < \alpha \leq 1$,  let
  $\alpha$-MaxCover be a variant of MaxCover for instances that
  satisfy the following conditions: If $p$ is a lower-bound on the
  frequencies of the elements and there are $m$ sets, then
  $\frac{p}{m} \geq \alpha$.
\end{definition}

\begin{theorem}\label{thm:ptas}
  For each $\alpha$, $0 < \alpha \leq 1$, there is a PTAS for
  $\alpha$-MaxCover.
\end{theorem}
\begin{proof}
  Fix some $\alpha$, $0 < \alpha \leq 1$.  Let $I = (N,\calS,K)$ be
  input instance of $\alpha$-MaxCover and let $\beta$ be our desired
  approximation ratio.  We let $m$ be the number of set in $\calS$ and
  $p$ be the lower bound on element frequencies. By definition, we
  have $\frac{p}{m} \geq \alpha$.  If $K > -\frac{m}{p}\ln(1 - \beta)$
  then we can run Algorithm~\ref{alg:greedy} and, by
  Theorem~\ref{theorem:greedy}, we obtain approximation ratio
  $\beta$. Otherwise, $K$ is bounded by a constant and enumerating all
  $K$-element subsets of $\calS$ gives a polynomial exact algorithm
  for the problem.~
\end{proof}

The exact complexity of $\alpha$-MaxCover is quite interesting.  Using
Algorithm~\ref{alg:greedy}, we show that it belongs to the second
level of Kintala and Fisher's $\beta$-hierarchy of limited
nondeterminism~\cite{fis-kin:j:beta}. In effect, it is unlikely that
the problem is $\np$-complete.

\begin{definition}[Kintala and Fisher~\cite{fis-kin:j:beta}]
  For each positive integer $k$, $\beta^k$ is the class of decision
  problems that can be solved in polynomial time, using additionally
  at most $O(\log^kn)$ nondeterministic bits (where $n$ is the size of
  the input instance).
\end{definition}
It is easy to see that $\beta^1$ is simply the class of problems
solvable in polynomial time; we can simulate $O(\log n)$ bits of
nondeterminism by trying all possible combinations. However, class
$\beta^2$ appears to be greater than $\p$ but smaller than $\np$ (of
course, since we do not know if $\p \neq \np$, this is only a
conjecture).

\begin{theorem}
  For each $\alpha$, $0 < \alpha < 1$, the decision variant of
  $\alpha$-MaxCover is in $\beta^2$.
\end{theorem}
\begin{proof}
  Fix some $\alpha$, $0 < \alpha < 1$.  We will give a
  $\beta^2$-algorithm for $\alpha$-MaxCover. Let $I = (N,\calS,K,T)$
  be an instance of $\alpha$-MaxCover (recall that $T$ is the number
  of elements we are required to cover). We let $p$ be the lower bound
  on elements' frequencies in $I$, we let $m = \|\calS\|$, and we let
  $n = \|N\|$. By definition, we have $\frac{p}{m} \geq
  \alpha$. W.l.o.g., we assume that $\|I\| \geq n+m$.

  Our algorithm works as follows. If $K > \frac{1}{\alpha}\ln(n)$
  then we run Algorithm~\ref{alg:greedy} and output its
  solution. Otherwise, we guess $K$ names of the sets from $\calS$ and
  check if these sets cover at least $T$ elements. If so, we accept
  and otherwise we reject on this computation path. 
  
  First, it is clear that the algorithm uses at most $O(\log^2|I|)$
  nondeterministic bits.  We execute the nondeterministic part of the
  algorithm only if $K < \frac{1}{\alpha}\ln(n) \leq
  \frac{1}{\alpha}\ln|I|$ and each set's name requires at most $\log
  m \leq \log |I|$ bits. Altogether, we use at most
  $O(\log^2|I|)$ bits of nondeterminism.
 
  Second, we need to show the correctness of the algorithm. Clearly,
  if the algorithm uses the nondeterministic part then certainly it
  finds an optimal solution. Consider then that the algorithm uses the
  deterministic part, based on Algorithm~\ref{alg:greedy}.  In this
  case we know that $K > \frac{1}{\alpha}\ln(n)$. Thus, the
  approximation ratio of Algorithm~\ref{alg:greedy} is greater than:
  $1 - e^{-\alpha K} > (1 - e^{-\ln n}) = 1 - \frac{1}{n}$. That is,
  the algorithm returns a solution that covers more than
  $\OPT(1-\frac{1}{n})$ elements and, since $\OPT \leq n$ and the
  number of covered elements is integer, the algorithm must find an
  optimal solution.
\end{proof}

\subsection{The MinNonCovered Problem}
In this section we considered the MinNonCovered problem, that is, a
version of MaxCover where the goal is to minimize the number of
elements left uncovered. In this case we give a randomized FPT
approximation scheme (presented as
Algorithm~\ref{alg:pApprovalDissat}).

Intuitively, the idea behind our approach is to extend a simple
bounded-search-tree algorithm for SetCover with upper-bounded
frequencies to the case of MaxCover. An FPT algorithm for SetCover
with frequencies upper-bounded by some constant $p$ could work
recursively as follows: If there still is some uncovered element $e$,
then nondeterministically guess one of the at-most-$p$ sets that
contain $e$ and recursively solve the smaller problem. The recursion
tree would have at most $K$ levels and $p^K$ leaves. The same approach
does not work directly for MaxCover because we do not know which
element $e$ to pick (in SetCover the choice is irrelevant because we
have to cover all the elements). However, it turns out that if we
choose $e$ randomly then, in expectation, we achieve a good result.

\begin{algorithm}[t!]
   \footnotesize
   \SetKwInput{KwParameters}{Parameters}
   \SetKwFunction{RecursiveSearch}{RecursiveSearch}
   \SetKwFunction{Main}{Main}
   \SetKwBlock{Block}
   \SetAlCapFnt{\footnotesize}
   \KwParameters{\\
          $\hspace{3pt}$ $(N,\calS,K)$ --- input MinNonCovered instance \\
          $\hspace{3pt}$ $p$ --- bound on the number of sets each element can belong to \\

          $\hspace{3pt}$ $\beta$ --- the required approximation ratio of the algorithm \\
          $\hspace{3pt}$ $\epsilon$ --- the allowed probability of achieving worse than $\beta$ approximation ratio
}
\vspace{3mm}	
    \RecursiveSearch{$s$, $\partialsol$}:
	\Block{
		\eIf{$s$ = $0$}
		{
			\Return{$\partialsol$}\;
		}
		{
			$e \leftarrow$ randomly select element not-yet covered by\\
							$\hspace{0.7cm}\partialsol$\;
			$\best \leftarrow \emptyset$\;
			\ForEach{$S \in \calS$ such that $e \in S$} 
			{
				$\sol \leftarrow$ \RecursiveSearch{$(s - 1)$, 
                    $\partialsol \cup \{S\}$}\;
				\If{$\sol$ {\rm\bf is better than} $\best$}
				{
					$\best \leftarrow \sol$\;
				}
			}
			\Return{$\best$}\;
		}
	}
	\hspace{5mm} \\
	\Main{}:
	\Block{
		$\best = \emptyset$\;
		\For{$i\leftarrow 1$ \KwTo $\left\lceil -\ln \epsilon/\left(\frac{\beta - 1}{\beta}\right)^K \right\rceil$}{
			$\sol$ = \RecursiveSearch{$K$, $\emptyset$}\;
			\If{$\sol$ {\rm\bf is better than} $\best$}
			{
				$\best \leftarrow \sol$\;
			}
		}
		\Return{$\best$}\;
	}

   \caption{\small The algorithm for the MinNonCovered problem with frequency upper bounded by $p$.}
   \label{alg:pApprovalDissat}
\end{algorithm}

\begin{theorem}\label{thm:pApprovalDissat}
  Algorithm~\ref{alg:pApprovalDissat} outputs a $\beta$-approximate
  solution for the MinNonCovered problem with probability $(1 -
  \epsilon)$. The time complexity of the algorithm is \[\poly(n,m) \cdot
  \left\lceil -\ln \epsilon/\left(\frac{\beta - 1}{\beta}\right)^K
  \right\rceil \cdot p^{K}\].
\end{theorem}
\begin{proof}
  Let $I = (N,\calS,K)$ be our input instance of the MinNonCovered
  problem and fix some $\beta$, $\beta > 1$, and $\epsilon$, $0 <
  \epsilon < 1$. Each element from $N$ appears in at most $p$ sets
  from $\calS$.

  By $p_{s}$ we denote the probability that a single invocation of the
  function \texttt{RecursiveSearch} (from the \texttt{Main} function)
  returns a $\beta$-approximate solution.  We will first show that
  $p_s$ is at least $\left(\frac{\beta - 1}{\beta}\right)^K$, and then
  we will invoke the standard argument that if we make
  $\left\lceil\frac{-\ln\epsilon}{p_s}\right\rceil$ calls to
  \texttt{RecursiveSearch}, then taking the best output gives a
  $\beta$-approximate solution with probability $(1-\epsilon)$.


  Let $\calC^*$ be some optimal solution for $I$, let $N^* \subseteq
  N$ be the set of elements covered by $\calC^*$, and let $U^* = N
  \setminus N^*$ be the set of the remaining, uncovered elements.
%
%
  Consider a single call to \texttt{RecursiveSearch} from the ``for''
  loop within the function \texttt{Main}.  Let $\Ev$ denote the event
  that during such a call, at the beginning of each recursive call, at
  least a $\frac{\beta-1}{\beta}$ fraction of the elements not covered
  by the constructed solution (i.e., the solution denoted
  $\partialsol$ in the algorithm) belongs to $N^*$.
  Note that if the complementary event, denoted $\overline{\Ev}$,
  occurs, then \texttt{RecursiveSearch} definitely returns a
  $\beta$-approximate solution. Why is this the case? Consider some
  tree of recursive invocations of \texttt{RecursiveSearch}, and some
  invocation of \texttt{RecursiveSearch} within this tree. Let $X$ be
  the number of elements not covered by $\partialsol$ at the beginning
  of this invocation.  If at most $\frac{\beta-1}{\beta}X$ of the
  not-covered elements belong to $N^*$, then---of course---the
  remaining at least $\frac{1}{\beta}X$ of them belong to $U^*$. In
  other words, then we have $\frac{1}{\beta}X \leq \|U^*\|$ and,
  equivalently, $X \leq \beta\|U^*\|$. This means that $\partialsol$
  already is a $\beta$-approximate solution, and so the solution
  returned by the current invocation of \texttt{RecursiveSearch} will
  be $\beta$-approximate as well.  (Naturally, the same applies to the
  solution returned at the root of the recursion tree.)


  Now, consider the following random process $\calP$. (Intuitively,
  $\calP$ models a particular branch of the \texttt{RecursiveSearch}
  recursion tree.)  We start from the set $N'$ of all the elements,
  $N' = N$, and in each of the next $K$ steps we execute the following
  procedure: We randomly select an element $e$ from $N'$ and if $e$
  belongs to $N^*$, we remove from $N'$ all the elements covered by
  the first\footnote{We assume that the sets in $\calC^*$ are ordered
    in some arbitrary way.} set from $\calC^*$ that covers $e$. Let
  $p_\opt$ be the probability that a call to \texttt{RecursiveSearch}
  (within \texttt{Main}) finds an optimal solution for $I$, and let
  $p_{\opt|\Ev}$ be the same probability, but under the condition that
  $\Ev$ takes place.  It is easy to see that $p_\opt$ is greater or
  equal than the probability that in each step $\calP$ picks an
  element from $N^*$.  Let $p_\hit$ be the probability in each step
  $\calP$ picks an element from $N^*$, under the condition that at the
  beginning of every step more than $\frac{(\beta - 1)}{\beta}$
  fraction of the elements in $N'$ belong to $N^*$. Again, it is easy
  to see that $p_{\opt|\Ev} \geq p_\hit$. Further, it is immediate
  to see that $p_\hit \geq \left(\frac{\beta - 1}{\beta}\right)^K$.

  Altogether, combining all the above fidnings, we know that the
  probability that \texttt{Re\-cur\-sive\-Search} returns a $\beta$-approximate
  solution is at most:
  \begin{align*}
    p_{s} \geq \probability(\overline{\Ev}) + \probability(\Ev)p_{\opt|\Ev}
    \geq p_{\opt|\Ev} \geq \left(\frac{\beta - 1}{\beta}\right)^K
    \textrm{.}
  \end{align*}
  (That is, either the event $\Ev$ does not take place and
  \texttt{RecursiveSearch} definitely returns a $\beta$-approximate
  solution, or $\Ev$ does occur, and then we lower-bound the
  probability of finding a $\beta$-approximate solution by the
  probability of finding the optimal one.)

  To conclude, the probability of finding a $\beta$-approximate
  solution in one of the $x = \left\lceil -\ln
    \epsilon/\left(\frac{\beta - 1}{\beta}\right)^K \right\rceil$
  independent invocations of \texttt{RecursiveSearch} from
  \texttt{Main} is at least:
  \begin{align*}
    1 - \left(1 - \left(\frac{\beta - 1}{\beta}\right)^K\right)^x \geq
    1 - e^{\ln\epsilon} = 1 - \epsilon \textrm{.}
  \end{align*}
  Establishing the running time of the algorithm is immediate, and so
  the proof is complete.~
\end{proof}

\begin{algorithm}[t]
   \small
   \SetAlCapFnt{\small}
   \SetKwFunction{Expon}{Expon}
   \KwParameters{\\$\hspace{3pt}$ $(N,\calS,K)$ --- input MaxCover instance\\
          $\hspace{3pt}$ $X$ --- the parameter of the algorithm\\
          $\hspace{3pt}$ $\calA(\cdot)$ --- an exact algorithm for MaxCover (returns the set 
                          of sets to be used in the cover)}
   $C = \{\}$\;
   \For{$i\leftarrow 1$ \KwTo $X$}{
      $\Cov \leftarrow \{e \in N: \exists_{S \in C} e \in S\}$ \;
      $S_{\best} \leftarrow \argmax_{S \in \{S_1, \dots, S_m\} \setminus C} \|\{e \in N \setminus \Cov: e \in S\}\|$\;
      $C \leftarrow C \cup \{S_{\best}\}$
   }
   $\uCov \leftarrow N \setminus \{e \in N: \exists_{S \in C} e \in S\}$ \;
   $C' \leftarrow \calA(\uCov, (K-X), S \setminus C)$ \;
   \Return{$C \cup C'$}
   \caption{\small An approximation algorithm for the unrestricted MaxCover problem.}
   \label{alg:greedyAndExpo}
\end{algorithm}

Algorithm~\ref{alg:pApprovalDissat} is very useful, especially in
conjunction with Algorithm~\ref{alg:pApproval}. The former one has to
provide a very good solution if it is possible to cover almost all the
elements and the latter one has to provide a very good solution if in
every solution many elements must be left uncovered.

\section{Algorithms for the Unrestricted Variant}
\label{sec:unrestricted}

So far we have focused on the MaxCover problem where element
frequencies were either upper- or lower-bounded. Now we consider the
completely unrestriced variant of the problem. In this case we give
exponential-time approximation schemes that, nonetheless, are not FPT.

The main idea, which is similar to that of
Cygan~et.~al~\cite{journals/ipl/CyganKW09} and of Croce and
Paschos~\cite{cro-pas:j:cover}, is to solve part of the problem using
an exact algorithm and to solve the remaining part using the greedy
algorithm (i.e., Algorithm~\ref{alg:greedy}). There are two possible
ways in which this idea can be implemented: Either we can first run
the exact algorithm and then solve the remaining part of the instance
using the greedy algorithm, or the other way round. We consider both
approaches, though a variant of the ``brute-force-first-then-greedy''
approach appears to be superior (at least as long as we do not have
exact algorithms that are significantly faster than a brute-force
approach).

We start with the analysis of Algorithm~\ref{alg:greedyAndExpo}, which
first runs the greedy part and then completes it using an exact
algorithm.

\begin{theorem}\label{thm:greedyAndExpo}
  Let $\calA$ be an exact algorithm for the MaxCover problem with time
  complexity $f(K,n, m)$.  For each instance $I = (N,\calS,K)$ of
  MaxCover and for each $X$, $0 \leq X \leq K$,
  Algorithm~\ref{alg:greedyAndExpo} returns an $\left(1 -
    \frac{X}{K}e^{-\frac{X}{K}}\right)$-approximate solution for $I$
  and runs in time $f(K-X,n, m) + \poly(n,K,m))$.
\end{theorem}
\begin{proof}
  Establishing the running time of the algorithm is immediate and,
  thus, below we focus on showing the approximation ratio.

  Let $I = (N,\calS,K)$ be an instance of MaxCover and let $X$ be an
  integer, $1 \leq X \leq K$.  We rename the elements in $\calS$ so
  that $\calS = \{S_1, \ldots, S_m\}$ and $S_1, \dots S_X$ are the
  consecutive elements selected in the first, greedy, ``for loop'' in
  Algorithm~\ref{alg:greedyAndExpo}.  For each $i$, $1 \leq i \leq m$,
  let $c_i = \|S_i \setminus (S_1 \cup \cdots \cup S_{i-1})\|$.  Let
  $N_{\OPT}$ denote the set of elements covered by some optimal
  solution and set $\OPT = \|N_{\OPT}\|$. Let $\Cov_i$ denote the set
  $S_1 \cup \cdots \cup S_{i-1}$. (That is, $\Cov_i$ is the set of
  elements in the variable $\Cov$ in the
  Algorithm~\ref{alg:greedyAndExpo} right before executing the $i$'th
  iteration of the ``for loop''. Of course, $\Cov_1 = \emptyset$.)
  Naturally, for each $i$, $1 \leq i \leq m$, we have $\|\Cov_i\| =
  \sum_{j=1}^{i-1}c_i$.

  We claim that for each $i$, $1 \leq i \leq X$, there exist $(K-i)$
  sets from $\calS \setminus \{S_1, \dots S_{i-1}\}$ that cover at
  least $\frac{K-i}{K}$ fraction of the elements from $N_{\OPT}
  \setminus \Cov_{i-1}$. Why is this the case? First, note that there
  are some $K$ sets from $\calS \setminus \{S_1, \dots S_{i-1}\}$ that
  cover $N_{\OPT} \setminus \Cov_{i-1}$ (it suffices to take the $K$
  sets from some optimal solution, if need be, replace those that
  belong to $\{S_1, \ldots, S_{i-1}\}$ with some arbitrarily chosen
  ones from $\calS \setminus \{S_1, \ldots, S_{i-1}\}$). Let $Q_1,
  \ldots, Q_K$ be these $K$ sets.  Consider some arbitrary assignment
  of the elements from $N_{\OPT} \setminus \Cov_{i-1}$ to the sets
  $Q_1, \ldots, Q_K$, such that each element is assigned to exactly
  one set. Further, consider an ordering of these sets according to
  the increasing number of assigned elements. If the $i$'th set in the
  ordering is assigned at most fraction $\frac{1}{K}$ of the elements,
  than each of the sets preceding the $i$'th one in the ordering also
  is assigned at most fraction $\frac{1}{K}$ of the elements. In
  consequence, the last $(K-i)$ sets from the ordering cover at least
  fraction $\frac{K-i}{K}$ of the elements. On the other hand, if the
  $i$'th set in the order is assigned more than fraction $\frac{1}{K}$
  of the elements then the following sets also are and, once again,
  the last $(K-i)$ elements cover at least fraction $\frac{K-i}{K}$ of
  the elements.

  In consequence, we see that for each $i$, $1 \leq i \leq X$, $c_i
  \geq \frac{1}{K}(\OPT - \sum_{j = 1}^{i-1} c_j)$. The reason is that
  since there are $K-i$ sets among $\calS \setminus \{S_1, \ldots,
  S_{i-1}\}$ that cover fraction $\frac{K-i}{K}$ of elements from
  $N_\OPT \setminus \Cov_i$, at least one of them must cover
  $\frac{1}{K}(\OPT - \|\Cov_i\|)$.  $S_i$ is chosen as a set that
  covers most sets from $N - \Cov_i$. It covers $c_i$ elements from $N
  - \Cov_i$, and, thus, $c_i \geq \frac{1}{K}(\OPT - \|\Cov_i\|) =
  \frac{1}{K}(\OPT - \sum_{j=1}^{i-1})$.

  We can now proceed with computing the algorithm's approximation
  ratio.  By the above reasoning, we observe that the solution
  provided by Algorithm~\ref{alg:greedyAndExpo} covers at least $c =
  \sum_{i=1}^{X} c_i + \frac{K-X}{K}(\OPT - \sum_{i = 1}^{X} c_i) =
  \frac{X}{K}\sum_{i=1}^{X} c_i + \frac{K-X}{K}\OPT$. Now, we assess
  the minimal value of $\sum_{i=1}^{X} c_i$. Minimization of
  $\sum_{i=1}^{X} c_i$ can be viewed as a linear programming task with
  the following constraints: for each $i$, $1 \leq i \leq X$, $c_i
  \geq \frac{1}{K}(\OPT - \sum_{j = 1}^{i-1} c_j)$.  Since we have $X$
  variables and $X$ constraints, we know that the minimum is achieved
  when each constraint is satisfied with equality (see, e.g.,
  \cite{Vazirani:2001:AA:500776}). Thus a solution to our linear
  program consists of values $c_{1, \min}, \ldots, c_{X,\min}$ that,
  for each $i$, $1 \leq i \leq X$, satisfy $c_{i, \min} =
  \frac{1}{K}(\OPT - \sum_{j = 1}^{i-1} c_{j, \min})$. By induction,
  we show that for each $i$, $1 \leq i \leq X$, $c_{i, \min} =
  \frac{1}{K}\left(\frac{K-1}{K}\right)^{i-1}\OPT$. Indeed, the claim
  is true for $i = 1$:
  \begin{align*}
    c_{1, \min} = \frac{1}{K}\OPT
  \end{align*}
  Now, assuming that $c_{i, \min} =
  \frac{1}{K}\left(\frac{K-1}{K}\right)^{i-1}\OPT$, we calculate
  $c_{(i+1), \min}$:
  \begin{align*}
    c_{(i+1), \min} &= \frac{1}{K}\left(\OPT - \sum_{j = 1}^{i} c_{j, \min}\right) \\
    &= \frac{1}{K}\OPT\left(1 - \frac{1}{K}\sum_{j = 1}^{i} \left(\frac{K-1}{K}\right)^{j-1}\right) \\
    &= \frac{1}{K}\OPT\left(1 - \frac{1}{K} \cdot \frac{1 - \left(\frac{K-1}{K}\right)^{i}}{1 - \left(\frac{K-1}{K}\right)}\right) \\
    &= \frac{1}{K}\OPT\left(1 - \left(1 - \left(\frac{K-1}{K}\right)^{i}\right)\right) \\
    &= \frac{1}{K}\OPT\left(\frac{K-1}{K}\right)^{i} \textrm{.}
  \end{align*}

  Thus we can lower-bound the number of elements covered by
  Algorithm~\ref{alg:greedyAndExpo} as follows:
  \begin{align*}
    c &= \frac{X}{K}\sum_{i=1}^{X} c_i + \frac{K-X}{K}\OPT \\
    &= \OPT\left(\frac{X}{K^2}\sum_{i=1}^{X}\left(\frac{K-1}{K}\right)^{i-1} + \frac{K-X}{K}\right) \\
    &= \OPT\left(\frac{X}{K^2} \cdot \frac{1 - \left(\frac{K-1}{K}\right)^{X}}{1 - \left(\frac{K-1}{K}\right)} + \frac{K-X}{K}\right) \\
    &= \OPT\left(\frac{X}{K}\left(1 - \left(\frac{K-1}{K}\right)^{X}\right) + \frac{K-X}{K}\right) \\
    &\geq \OPT\left(1 - \frac{X}{K}e^{-\frac{X}{K}}\right) \textrm{.}
  \end{align*}
  This completes the proof.~
\end{proof}

The idea of the proof of Theorem~\ref{thm:greedyAndExpo} is simillar
to the algorithm of Cygan~et.~al~\cite{journals/ipl/CyganKW09} for the
problem of weighted set cover. Theorem~\ref{thm:greedyAndExpo} gives a
good-quality result provided we knew an optimal algorithm with the
better complexity than exhaustive search.
Otherwise, we can obtain even better results using
Algorithm~\ref{alg:expoAndGreedy}, which first runs a brute-force
approach and completes it using the greedy algorithm.

\begin{algorithm}[t]
   \small
   \SetAlCapFnt{\small}
   \SetKwFunction{better}{better}
   \KwParameters{\\$\hspace{3pt}$ $(N,\calS,K)$ --- input MaxCover instance\\
          $\hspace{3pt}$ $X$ --- the parameter of the algorithm}
   $C = \{\}$\;
   $C_{best} = \{\}$\;
   \ForEach{$(K-X)$-element subset $C$ of $\calS$}{
      \For{$i\leftarrow (K-X+1)$ \KwTo $K$}{
          $\Cov \leftarrow \{e \in N: \exists_{S \in C} e \in S\}$ \;
          $S_{\best} \leftarrow \argmax_{S \in \{S_1, \dots, S_m\} \setminus C}$ 
               $\{e \in N \setminus \Cov: e \in S\}\|$\;
          $C \leftarrow C \cup \{S_{\best}\}$
      }
      $C_{best} \leftarrow$ better solution among $C_{best}$ and $C$\;
   }
   \Return{$C_{best}$}
   \caption{\small The approximation algorithm for the MaxCover problem.}
   \label{alg:expoAndGreedy}
\end{algorithm}

\begin{theorem}
  For each instance $I = (N,\calS,K)$ of MaxCover and each integer
  $X$, $0 \leq X \leq K$, Algorithm~\ref{alg:expoAndGreedy} computes
  an $\left(1 - \frac{X}{K}e^{-1}\right)$-approximate solution for
  $I$ in time ${m \choose K-X} + \poly(K,n,m)$.
\end{theorem}
\begin{proof}
  Let $I = (N,\calS,K)$ be our input instance and let $\calC^{*}$,
  $\calC^{*} \subseteq \calS$, denote some optimal solution. Let
  $\calC_X^{*}$ denote a subset of $(K-X)$-elements from $\calC^{*}$
  that together cover the greatest number of the elements. Thus the
  sets from $\calC_X^{*}$ cover at least a fraction $\frac{K-X}{K}$ of all
  the elements covered by the optimal solution. Consider the problem
  of covering the elements uncovered by $\calC_X^{*}$ with $X$ sets
  from $(\calS \setminus \calC_X^{*})$. We know that $(\calC^{*} \setminus
  \calC_X^{*}$ is an optimal solution for this problem. On the other
  hand, we also know that the greedy algorithm achieves approximation
  ratio $(1 - \frac{1}{e})$ for the problem. Thus, the approximation
  ratio for the original problem is:
  \begin{align*}
    \left(\frac{K-X}{K} + \frac{X}{K}\left(1 - \frac{1}{e}\right)
    \right) = \left(1 - \frac{X}{K}e^{-1}\right) \textrm{.}
  \end{align*}
  It is immediate to establish the running time of the algorithm and
  so the proof is complete.~
\end{proof}

If we wish to solve MaxVertexCover rather than MaxCover, then in
Algorithm~\ref{alg:expoAndGreedy} we should replace the greedy
approximation algorithm with that of Ageev and
Sviridenko~\cite{age-svi:b:covers}.

\begin{corollary}\label{alg5:maxvertexcover}
There exists an $\left(1 - \frac{X}{4K}\right)$-approximation algorithm for MaxVertexCover problem running in time ${m \choose K-X} + \poly(K,n,m)$
\end{corollary}

It is quite evident that as long as algorithm $\calA$ used within
Algorithm~\ref{alg:greedyAndExpo} is the simple brute-force algorithm
that tries all possible solutions, then
Algorithm~\ref{alg:expoAndGreedy} is superior; in the same time it
achieves a better approximation ratio. It turns out that, for the case
of MaxVertexCover, Algorithm~\ref{alg:expoAndGreedy} (in the variant
from Corollary~\ref{alg5:maxvertexcover}) is also better than the
algorithm of Croce and
Paschos~\cite{cro-pas:j:cover}.\footnote{Algorithm~\ref{alg:greedyAndExpo}
  cannot be directly compared to the algorithm of Croce~and
  ~Paschos~\cite{cro-pas:j:cover} for the following reason.
  Algorithm~\ref{alg:greedyAndExpo} uses specifically a greedy
  algorithm which is the best known approximation algorithm for
  MaxCover, but which is suboptimal for MaxVertexCover. In contrast,
  the algorithm of Croce~and ~Paschos~\cite{cro-pas:j:cover} can use,
  e.g., the $\frac{3}{4}$-approximation algorithm of Ageev and
  Sviridenko~\cite{age-svi:b:covers}. One could, of course, try to use
  the algorithm of Ageev and Sviridenko in
  Algorithm~\ref{alg:greedyAndExpo}, but our analysis does not work
  for this case.}

The idea behind the algorithm of Croce~and
~Paschos~\cite{cro-pas:j:cover} for MaxVertexCover is similar to that
behind our Algortihm~\ref{alg:expoAndGreedy}.  Specifically, given two
algorithms for MaxVertexCover, approximation algorithm $\calA_a$ and
exact algorithm $\calA_e$, for a given value $X$ it first uses
$\calA_{e}$ to find am optimal solution that uses $K-X$ vertices (out
of the $K$ vertices that we are allowed to use in the full solution),
then it remeves these $K-X$ vertices and solves the remaining part of
the problem using $\calA_e$.  Assuming that $\beta_{a}$ is the
approximation ratio of the algorithm $\calA_{a}$, this approach
results in the approximation ratio equal to $\left(\frac{X}{K} +
  \beta_{a}\left(1 - \frac{X}{K}\right)^2\right)$.

Below we compare Algorithm~\ref{alg:expoAndGreedy} (version from
Corollary~\ref{alg5:maxvertexcover}) with the algorithm of Croce and
Paschos~\cite{cro-pas:j:cover}. As the components $\calA_a$ and
$\calA_e$ we use, respectively, the $\frac{3}{4}$-approximation
algorithm of Ageev and Sviridenko~\cite{age-svi:b:covers} and the
brute-force algorithm that tries all possible solutions.  The best
known exact algorithm for MaxVertexCover is due to
Cai~\cite{cai:j:cardinality-constrained} and has the complexity
$O(m^{0.792K})$, but this algorithm uses exponential amount of space.
Since exponential space complexity might be much less practical than
exponential time complexity, we decided to use the brute-force
approach (to the best of our knowledge there, there is no better exact
algorithm running in a polynomial space).
\begin{figure}[tb]
  \begin{center}
    \includegraphics[scale=0.5]{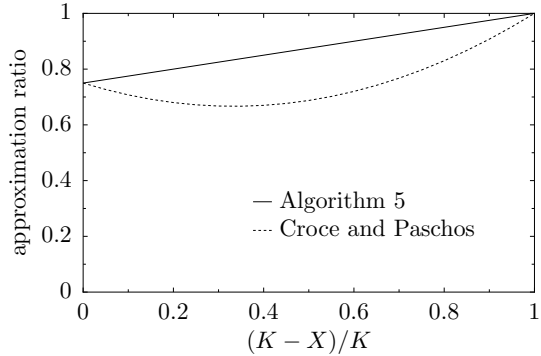}
  \end{center}
  \vspace{-0.75cm}
  \caption{The comparison of the approximation ratios of Algorithm~\ref{alg:expoAndGreedy} and the algorithm of Croce~and~Paschos~\cite{cro-pas:j:cover} for MaxVertexCover.}
  \label{fig:approximationComparison}
\end{figure}
We present our comparison in Figure~\ref{fig:approximationComparison}.
The $x$-axis represents the parameter $\frac{K-X}{K}$, measuring the
fraction of the solution obtained using the exact algorithm (for $0$
we use the approximation algorithm alone and for $1$ we use the exact
algorithm alone). On the $y$-axis we give approximation ratio of each
algorithm.  In other words, for each point on the $x$-axis we set the
$X$ parameters of the algorithms to be equal, so that their running
times are the same, and we compare their approximation guarantees.

We conclude that, as long as we use the brute-force algorithm as the
exact one, Algorithm~\ref{alg:expoAndGreedy} gives considerably better
approximation guarantees than that of Croce and
Paschos. Figure~\ref{fig:approximationComparison} also exposes one
potential weakness of the algorithm of Croce and Paschos. Apparently,
for some cases increasing the complexity of the algorithm results in
the decrease of its approximation guarantee.

It is quite interesting to understand the reasons behind the differing
performance of Algorithm~\ref{alg:expoAndGreedy} and that of Croce and
Paschos. In some sense, the algorithms are very similar. If we use the
brute-force algorithm as the exact one in the algorithm of Croce and
Paschos, then the main difference is that our algorithm runs the
approximation algorithm for each possible solution tried by the
brute-force algorithm, and Croce and Paschos's algorithm only runs the
approximation algorithm once, for the best partial solution. In
effect, our algorithm can exploit situations where it is better when
the exact algorithm does not find an optimal solution for the
subproblem, but rather leaves ground for the approximation algorithm
to do well. Naturally, such strategy is only possible if we have
additional knowledge of the structure of the exact algorithm (here,
the brute-force algorithm). The result of Croce and Paschos pays the
price for being more general and being able to use any combination of
the approximation algorithm and the exact algorithm.

\section{Conclusions}

Motivated by the study of winner-determination under
Chamberlin--Courant's voting rule (with approval misrepresentation),
we have considered the MaxCover problem with bounded frequencies and
its minimization variant, the MinNonCovered problem, from the point of
view of approximability by FPT algorithms. We have shown that for
upper-bounded frequencies there is an FPT approximation scheme for
MaxCover and a randomized FPT approximation scheme for
MinNonCovered. For lower-bounded frequences we have shown that the
standard greedy algorithm for MaxCover may achieve a better
approximation ratio than in the unrestricted case. Finally, we have
shown that in the unrestricted case there are good exponential-time
approximation algorithms (though, not FPT ones) that combine exact and
greedy algorithms and smoothly exchange the quality of the
approximation for the running time.
Some of our results regarding MaxCover with bounded frequencies
improve previously known results for MaxVertexCover. In particular,
our Algorithm~\ref{alg:pApproval} improves upon the approximation
scheme given by Marx, and our Algorithm~\ref{alg:expoAndGreedy}
improves upon the result of Croce and Paschos~\cite{cro-pas:j:cover}
(provided we use brute-force algorithm as the underlying exact
algorithm in the scheme proposed by Croce and Paschos; this is
reasonable if we are interested in algorithms that use only polynomial
amount of space).

There are several interesting directions for future research. For
example, is it possible to obtain FPT approximation schemes for
MaxCover with lower-bounded element frequencies? Further, what is the
exact complexity of MaxCover (with or without lower-bounded
frequencies)? We have quickly observed its $\wtwo$-hardness, but does
it belong to $\wtwo$? (It is quite easy, however, to show that it
belongs to $\wpclass$.) We are also interested in the exact complexity
of MaxCover with lower-bounded frequencies for the case where we
require the ratio of frequency lower-bound and the number of sets to
be at least some given value $\alpha$, $0 < \alpha < 1$? We have given
a PTAS for this variant of the problem (see Theorem~\ref{thm:ptas})
and have shown its membership in $\beta^2$, but we did not attempt to
prove its completeness for any particular complexity class.





\bibliographystyle{abbrv}
\bibliography{main}

\end{document}